\documentclass{article}
\usepackage{amsmath, amssymb, amsthm}
\usepackage{setspace}
\usepackage[a4paper, margin = 2.54cm]{geometry}
\usepackage{authblk}
\usepackage{graphicx}
\usepackage{fancyhdr}
\newtheorem{theorem}{Theorem}[subsubsection]
\newtheorem{lemma}[theorem]{Lemma}

\theoremstyle{remark}

\numberwithin{equation}{section}
\numberwithin{theorem}{section}
\numberwithin{figure}{section}

\DeclareMathOperator{\supp}{supp}
\DeclareSymbolFont{bbold}{U}{bbold}{m}{n}
\DeclareSymbolFontAlphabet{\mathbbold}{bbold}
\begin{document}
\title{A Lower Bound on the Renormalized Nelson Model}
\author{Gonzalo A. Bley}
\affil{Institut for Matematik, Aarhus Universitet, \authorcr Ny Munkegade 118,  8000 Aarhus C, Denmark}
\maketitle
\begin{abstract}
We provide explicit lower bounds for the ground-state energy of the renormalized Nelson model in terms of the coupling constant $\alpha$ and the number of particles $N$, uniform in the meson mass and valid even in the massless case. In particular, for any number of particles $N$ and large enough $\alpha$ we provide a bound of the form $-C\alpha^2 N^3\log^2(\alpha N)$, where $C$ is an explicit positive numerical constant; and if $\alpha$ is sufficiently small, we give one of the form $-C\alpha^2 N^3\log^2 N$ for $N \geq 2$, and $-C\alpha^2$ for $N = 1$. Whereas it is known that the renormalized Hamiltonian of the Nelson model is bounded below (as realized by E. Nelson) and implicit lower bounds have been given elsewhere (as in a recent work by Gubinelli, Hiroshima, and L\"{o}rinczi), ours seem to be the first fully explicit lower bounds with a reasonable dependence on $\alpha$ and $N$. We emphasize that the logarithmic term in the bounds above is probably an artifact in our calculations, since one would expect that the ground-state energy should behave as $-C\alpha^2 N^3$ for large $N$ or $\alpha$, as in the polaron model of H. Fr\"{o}hlich.
\end{abstract}
\begin{section}{Introduction}
The Nelson Model describes the interaction of non-relativistic nucleons with a meson field. The model is ascribed to Edward Nelson, who presented it in some detail for the first time in a conference held in 1963 \cite{N1}. In a subsequent article that was published shortly afterward, Nelson \cite{N2} studied the model in a substantially more organized and systematic way: instead of using stochastic-analytic methods as in \cite{N1}, which led to rather unwieldy computations, he used operator-theoretic techniques, which made the arguments substantially more transparent. The Hamiltonian describing the interaction between nucleons and mesons is given by
\begin{equation}
H_{\alpha, \mu}^{N, \Lambda} \equiv -\sum_{n = 1}^N\frac{\Delta_n}{2} + \int_{\mathbb{R}^3}\chi_{\Lambda}(k)\omega(k)a_k^{\dagger}a_k\,dk + \sqrt{\alpha}\sum_{n = 1}^N\int_{\mathbb{R}^3}\frac{\chi_{\Lambda}(k)}{\sqrt{\omega(k)}}(e^{ikx_n}a_k + e^{-ikx_n}a_k^{\dagger})\,dk,
\label{equation.nelson.hamiltonian}
\end{equation}
acting on $L^2(\mathbb{R}^{3N})\otimes\mathcal{F}$, where $\mathcal{F}$ is the Fock space over $L^2(\mathbb{R}^3)$. The symbols appearing in \eqref{equation.nelson.hamiltonian} are defined as follows: $N$ is the total number of nucleons; if a vector $x$ in $\mathbb{R}^{3N}$ is written as $(x^1, x^2, \ldots , x^N)$, with $x^i \equiv (x^i_1, x^i_2, x^i_3)$, then $\Delta_n$ is the operator $\partial^2\! / \partial^2 x_1^n + \partial^2\! / \partial^2 x_2^n + \partial^2\! / \partial^2 x_3^n$; $\Lambda$ is a non-negative real number that plays the role of an ultraviolet cutoff; $\chi_{\Lambda}(k)$ is the indicator function of the set $\left\{ k : |k| \leq \Lambda\right\}$; $\alpha$ is a number greater than or equal to zero and measures the strength of the nucleon-field interaction; $\omega(k)$ is the function $\sqrt{k^2 + \mu^2}$, where $\mu$ is a non-negative number measuring the mass of the mesons; and $a_k$, $a_k^{\dagger}$ are the standard annihilation and creation operators for the Fock space $\mathcal{F}$, satisfying $[a_k, a_{k'}^{\dagger}] = \delta(k - k')$ and $[a_k, a_{k'}] = [a_k^{\dagger}, a_{k'}^{\dagger}] = 0$.

We shall now list a series of spectral properties of $H_{\alpha, \mu}^{N, \Lambda}$. We shall refer to $H_{\alpha, \mu}^{N, \Lambda}$ simply as $H$, whenever there is no risk of confusion. It is known that for finite $\Lambda$, $H$ is self-adjoint and bounded-below, but as $\Lambda$ goes to $\infty$ the infimum of the spectrum goes to $-\infty$ \cite{N2}. Despite this fact, a stabilizing term can be added to the Hamiltonian, which allows one to take $\Lambda$ to $\infty$ while still retaining in the limit a bounded-below operator. This was first discovered by Nelson in \cite{N1, N2}. The precise term is
\begin{equation}
Q_{\alpha, \mu}^{N, \Lambda} \equiv \alpha N \int\frac{\chi_{\Lambda}(k)\,dk}{\omega(k)\left[k^2/2 + \omega(k)\right]},
\label{equation.renormalizing.term}
\end{equation}
and the result is that for all real $t$, $\exp[it(H_{\alpha, \mu}^{N, \Lambda} + Q_{\alpha, \mu}^{N, \Lambda})] \to \exp(it\widehat{H}_{\alpha, \mu}^N)$ strongly as $\Lambda \to \infty$, for some self-adjoint, bounded-below operator $\widehat{H}$. $\widehat{H}$ is known as the renormalized Hamiltonian of the Nelson model. The term $Q$, defined in \eqref{equation.renormalizing.term}, appears naturally after performing a Gross transformation on $H$; we refer the reader to \cite{N2} for more information. The main purpose of this article is to give a numerical lower bound to the infimum of the spectrum of $\widehat{H}$ with an explicit dependence on $N$, $\alpha$, and $\mu$. To the best of our knowledge, at least one implicit lower bound has appeared before, in a work by Gubinelli, Hiroshima and L\"{o}rinczi \cite[Corollary 2.18]{GHL}. Ours seems to be the first explicit one.

Our main result consists of a family of lower bounds for the renormalized Nelson model Hamiltonian $\widehat{H}$. The general result we provide later below is however rather complicated, and it is perhaps better in this introduction to simply state particular bounds that may be derived from it. For example, we will show below that for sufficiently large $\alpha$ and all $\Lambda, \mu$, and $N$,
\begin{equation}
E_{\alpha, \mu}^{N, \Lambda} + Q_{\alpha, \mu}^{N, \Lambda} \geq -C\alpha^2 N^3\log^2\left(\alpha N\right),
\label{equation.main.result}
\end{equation}
where $E_{\alpha, \mu}^{N, \Lambda}$ is the ground-state energy of $H_{\alpha, \mu}^{N, \Lambda}$ and $C$ is an explicit positive number; and for small enough $\alpha$ we will find a lower bound of the form $E_{\alpha, \mu}^{N, \Lambda} + Q_{\alpha, \mu}^{N, \Lambda} \geq -C\alpha^2 N^3\log^2 N$ when $N \geq 2$, and $E_{\alpha, \mu}^{N, \Lambda} + Q_{\alpha, \mu}^{N, \Lambda} \geq -C\alpha^2$ when $N = 1$. The ground-state energy of the renormalized Hamiltonian is then obtained on the left side of the inequalities by taking $\Lambda \to \infty$ (see, for example, the discussion preceding \cite[Equation (1.2)]{GM}). We note that the constants $C$ are independent of all the model parameters ($\alpha$, $\mu$, $\Lambda$, and $N$), and that the results are valid for all meson masses, including the massless case $\mu = 0$; this, in spite that there is no ground state for the massless Hamiltonian in the absence of an infrared cutoff \cite{LMS}. We are inclined to believe that the logarithmic factors above are just artifacts; after all, an $N^3$ behavior of the ground-state energy for large $N$ is something one can find in other cases, such as the Fr\"{o}hlich polaron model \cite{Fr} (see \cite{BB1} and references therein, and also Section \ref{section.polaron} in this article), and when considering gravitating particles in stars \cite[Chapter 13]{LS}. We do not see any compelling reason to believe that the Nelson model could give rise to something new and different, such as an $N^3\log N$ behavior.

The idea behind the proof is as follows. The ground-state energy of $H$ (with a finite value of $\Lambda$) may be bounded from below by means of a Feynman-Kac-like formula
\begin{equation}
\text{inf spec }H \geq -\limsup_{T \to \infty}T^{-1}\log\left\{\sup_{x \in \mathbb{R}^{3N}}E^x\left[\exp\left(\mathcal{A}_T\right)\right]\right\},
\label{equation.ground.state.energy}
\end{equation}
for some Brownian functional $\mathcal{A}_T$ (an $L^2$ functional of $3N$-dimensional Brownian paths on $[0, T]$) that may be written explicitly in the case of the Nelson model after integrating the meson field variables; this will be explained later below. $E^x$ denotes expectation with respect to Brownian motion starting at the point $x \in \mathbb{R}^{3N}$. We then use an idea appearing in a recent article by the author and L.E. Thomas \cite{BT}, namely we apply the Clark-Ocone formula, well-known from Malliavin Calculus and Mathematical Finance \cite{Nu, KS2, KS1}, which allows one to rewrite this action $\mathcal{A}_T$ as
\begin{equation}
\mathcal{A}_T = E^x\left(\mathcal{A}_T\right) + \int_0^T\!\rho_t\,dX_t,
\label{equation.clark.ocone.formula}
\end{equation}
for some unique, adapted, $\mathbb{R}^{3N}$-valued, $L^2$ process $\rho$, where $X$ is $3N$-dimensional Brownian motion, and the integral is It\^{o}. We remark two things about equation \eqref{equation.clark.ocone.formula}: it is valid for any Brownian functional $\mathcal{A}_T$, and the point $x$ defines the Brownian paths $\mathcal{A}_T$ acts on, meaning that a change in $x$ changes $\mathcal{A}_T$ (and also in general $\rho$). A simple supermartingale estimate allows one then to obtain the bound
\begin{equation}
E^x\left[\exp\left(\mathcal{A}_T\right)\right] \leq \exp\left[E^x\left(\mathcal{A}_T\right)\right]E^x\left[\exp\left(\frac{p^2}{2(p - 1)}\int_0^T\!\rho_t^2\,dt\right)\right]^{1 - 1/p}
\label{equation.clark.ocone.estimate}
\end{equation}
for any $p > 1$ \cite{BT}. Except for some factors and an exponent, the right side in \eqref{equation.clark.ocone.estimate} is identical to the left, with a new action $\int_0^T\rho_t^2\,dt$, which in the case of the Nelson model will turn out to be more tractable than the original one, $\mathcal{A}_T$. We will then use the results in \cite{BT} to further bound this new action. For every $x$, the factor $E^x(\mathcal{A}_T)$ in \eqref{equation.clark.ocone.estimate} will happen to be, up to sublinear terms (growing more slowly than $T$ as $T \to \infty$), exactly $TQ$, where $Q$ is defined above, in equation \eqref{equation.renormalizing.term}, which will at the end allow us to obtain the desired bound. This is something to take note of, and should not be overlooked: the fact that $E^x(\mathcal{A}_T) = TQ + o(T)$ for every $x$ is remarkable, because in some sense it indicates that performing a Clark-Ocone expansion at the functional-integral level amounts to doing the Gross transformation that Nelson considered in the context of operator methods.

The structure of the rest of the article is as follows. In Section \ref{section.proof} a proof of the general lower bound for the renormalized Nelson model Hamiltonian (and in particular of the large- and small-$\alpha$ cases stated above) is provided. In Section \ref{section.polaron} similar calculations are performed for the multi-particle polaron model (which is described there), and a lower bound for its spectrum is obtained. We obtain an answer consistent with a simple upper bound, $-0.109\alpha^2N^3$ \cite{BB1}, which essentially follows from an argument of Pekar \cite{P}. At the end an appendix is provided, where the main Feynman-Kac-like formula \eqref{equation.ground.state.energy} is derived, and also a combinatorial argument is given, that allows for an improvement on the estimate on the ground-state energy of the Nelson model.

The present article is partly based on the Ph.D. thesis of the author \cite{BPT}, and we will refer to it from time to time. We would like to thank Oliver Matte, Jacob Schach M\o ller, and Lawrence Thomas for productive discussions. The author also acknowledges partial support from the Danish Council for Independent Research (Grant number DFF-4181-00221).
\end{section}
\begin{section}{Proof of the Lower Bound for the Nelson Model}
\label{section.proof}
We provide in this section the proof of the lower bound mentioned in the introduction for the Nelson model, with Hamiltonian \eqref{equation.nelson.hamiltonian}. After integrating the field variables, one encounters the following expression for $\mu > 0$:
\begin{align}
& \text{inf spec }H_{\alpha, \mu}^{N, \Lambda}\nonumber\\
\geq & -\limsup_{T \to \infty}T^{-1}\log\left\{\sup_{x \in \mathbb{R}^{3N}}E\left[\exp\left(\alpha\sum_{m, n}
\int\!\!\!\int_0^T\!\!\!\int_0^t\frac{\chi_{\Lambda}(k)e^{-\omega(k)(t - s)}}{\omega(k)}e^{-ik(X_t^m - X_s^n + x^m - x^n)}\,ds\,dt\,dk\right)\right]\right\},
\label{equation.functional.integral.nelson}
\end{align}
a derivation of which appears in the appendix. (The massless case will be addressed later in this section.) In \eqref{equation.functional.integral.nelson} we have passed the $x$-dependence to the functional $\mathcal{A}_T$, and for this reason $E$ denotes $E^0$ (a practice we shall continue to use in the rest of the article). We then apply the Clark-Ocone formula, equation \eqref{equation.clark.ocone.formula}. The requirement for that formula to be valid is that the action $\mathcal{A}_T$, defined here as the argument of the exponential in \eqref{equation.functional.integral.nelson},
\begin{gather}
\mathcal{A}_T \equiv \alpha\sum_{m, n}
\int\!\!\!\int_0^T\!\!\!\int_0^t\frac{\chi_{\Lambda}(k)e^{-\omega(k)(t - s)}}{\omega(k)}e^{-ik(X_t^m - X_s^n + x^m - x^n)}\,ds\,dt\,dk,
\label{equation.nelson.model.action}
\end{gather}
be real and in $L^2$. Real-valuedness is immediate, as the sine function is odd, and integration over all $k$'s makes the imaginary part of the $k$-integral vanish. Square integrability follows from a simple estimate on a single summand of $\mathcal{A}_T$,
\begin{gather}
\left|\int\!\!\!\int_0^T\!\!\!\int_0^t\frac{\chi_{\Lambda}(k)e^{-\omega(k)(t - s)}}{\omega(k)}e^{-ik(X_t^m - X_s^n + x^m - x^n)}\,ds\,dt\,dk\right| \leq \int\!\!\!\int_0^T\!\!\!\int_0^t\frac{\chi_{\Lambda}(k)e^{-\omega(k)(t - s)}}{\omega(k)}\,ds\,dt\,dk,
\end{gather}
which is independent of the Brownian path chosen. The first part of the proof will be computing $E(\mathcal{A}_T)$.
\begin{subsection}{Computation of $E(\mathcal{A}_T)$}
\label{subsection.expectation}
The calculation of the expectation of $\mathcal{A}_T$ comes about after using the independence of Brownian motions corresponding to different particles and the fact that $E(e^{ikZ}) = e^{-k^2\sigma^2/2}$ for a 3D Gaussian random variable $Z$ with zero mean and variance $\sigma^2$,
\begin{align}
E(\mathcal{A}_T) = \, & \alpha N T\int\frac{\chi_{\Lambda}(k)}{\omega(k)\left[k^2/2 + \omega(k)\right]}\,dk - \alpha N\int\frac{\chi_{\Lambda}(k)}{\omega(k)\left[k^2/2 + \omega(k)\right]^2}\left(1 - e^{-\left[k^2/2 + \omega(k)\right]T}\right)\,dk\nonumber\\
& + \alpha N (N - 1)\int_0^T\!\!\!\int_0^t\!\!\int\frac{e^{-k^2t/2}e^{-k^2s/2}}{\omega(k)}e^{-ik(x^m - x^n)}e^{-\omega(k)(t - s)}\chi_{\Lambda}(k)\,dk\,ds\,dt.
\end{align}
The first term divided by $T$ is, as has been said already in the introduction, the renormalizing term \eqref{equation.renormalizing.term}, which we have called $Q_{\alpha, \mu}^{N, \Lambda}$. The second term is certainly sublinear, and so may be neglected. The third term also turns out to be sublinear. This can be proven in a rather direct way: By using the simple estimates $\omega(k) \geq |k|$ and $\chi_{\Lambda} \leq 1$, we can bound the third term from above as
\begin{align}
& \alpha N(N - 1)\int_0^T\!\!\!\int_0^t\!\!\int\frac{e^{-k^2t/2}e^{-k^2s/2}}{|k|}e^{-|k|(t - s)}\,dk\,ds\,dt\nonumber\\
= & \, 4\pi\alpha N(N - 1)\int_0^T\!\!\!\int_0^t\!\!\int_0^{\infty} e^{-r^2 t/2}e^{-r^2 s/2}e^{-r(t - s)}r \,dr\,ds\,dt \equiv 4\pi\alpha N(N - 1)I.
\end{align}
After computing the $s$-integral we find
\begin{gather}
I = \int_0^{\infty}\!\!\!\int_0^T\frac{e^{-r^2 t} - e^{-r(1 + r/2)t}}{1 - r/2}\,dt\,dr,
\end{gather}
and then pick a number $0 < \varepsilon < 2$ and split $I$ as
\begin{align}
& \,\int_0^{\varepsilon}\int_0^T \frac{e^{-r^2 t}\left(1 - e^{-r(1 - r/2)t}\right)}{1 - r/2}\,dt\,dr + \int_{\varepsilon}^{2}\int_0^T \frac{e^{-r^2 t}\left(1 - e^{-r(1 - r/2)t}\right)}{1 - r/2}\,dt\,dr\nonumber\\
& \, + \int_{2}^{\infty}\int_0^T \frac{e^{-r\left(1 + r/2\right)t}\left(1 - e^{-r(r/2 - 1)t}\right)}{r/2 - 1}\,dt\,dr\nonumber\\
\leq & \, \int_0^{\varepsilon}\int_0^T(1 - r/2)^{-1}\,dt\,dr + \int_{\varepsilon}^2\int_0^T e^{-r^2 t} rt\,dt\,dr + \int_2^{\infty}\int_0^T e^{-r(1 + r/2)t}rt\,dt\,dr\nonumber\\
\leq & \, -2T\log(1 - \varepsilon/2) + 3\varepsilon^{-2}.
\end{align}
By picking then $1 - \varepsilon/2 = e^{-T^{-1/4}}$, we obtain then that $I \leq CT^{3/4}$ for large enough $T$ and some constant $C$.

The upshot is that $E(\mathcal{A}_T) = Q_{\alpha, \mu}^{N, \Lambda}T + o(T)$. We now proceed to compute and estimate the second term in the Clark-Ocone expansion, namely $\int_0^T\!\rho_t\,dX_t$.
\end{subsection}
\begin{subsection}{Computation of $\int_0^T\!\rho_t\,dX_t$}
In equation \eqref{equation.clark.ocone.formula} we stated what the Clark-Ocone formula says in general, but no procedure was given to compute $\rho$. In the case of interest here, the Nelson model action, equation \eqref{equation.nelson.model.action}, the process $\rho$ can be calculated explicitly as $E\left(D_t \mathcal{A}_T | \mathcal{F}_t\right)$, where $\mathcal{F}_t$ is the standard filtration for $3N$-dimensional Brownian motion $X_t$, namely $\sigma\left(X_s : 0 \leq s \leq t\right)$, and $D_t$ is the Malliavin derivative of $\mathcal{A}_T$. An introduction to the Malliavin derivative here would take us too far afield -- the reader is simply referred to the publication of the author and Thomas \cite{BT} already mentioned in the introduction, where a fairly complete introduction to this operator is introduced, discussed, and applied in certain cases of interest. Enough for our purposes here will be to provide the following prescription for computing this derivative for a certain class of Brownian functionals: Let $f : \mathbb{R}^m \to \mathbb{R}$ be a smooth function with polynomial growth, and $g_1, g_2, \ldots , g_m$ a collection of functions $[0, T] \to \mathbb{R}^n$ in $L^2$. Then, if $X$ represents $n$-dimensional Brownian motion and $W(h)$ is the It\^{o} integral of an $L^2$ function $h$ from $[0, T]$ to $\mathbb{R}^n$, namely $\int_0^T h_t\,dX_t$,
\begin{gather}
D_t f\left[W(g_1), W(g_2), \ldots, W(g_m)\right] = \sum_{i = 1}^m\frac{\partial f}{\partial x_i}\left[W(g_1), W(g_2), \ldots, W(g_m)\right]g_i(t).
\label{equation.malliavin.prescription}
\end{gather}
We shall only need \eqref{equation.malliavin.prescription} for $m = 1$, and therefore we will record such result here for future reference:
\begin{gather}
D_t f\left[W(g)\right] = f'\!\left[W(g)\right] g(t).
\label{equation.malliavin.prescription.single}
\end{gather}

Even though \eqref{equation.malliavin.prescription.single} does not apply immediately to the current case, action \eqref{equation.nelson.model.action}, it can easily be adapted to it, in particular by approximating the integrals by Riemann sums. (See \cite{BT} for more information.) By noting that each one of the summands of $\mathcal{A}_T$, defined in \eqref{equation.nelson.model.action}, is actually real (the imaginary part of the $k$-integral vanishing), one arrives at the end at the formula
\begin{align}
D_u\mathcal{A}_T = \, & D_u\left(\alpha\sum_{m, n}\int\!\!\!\int_0^T\!\!\!\int_0^t\frac{\chi_{\Lambda}(k)e^{-\omega(k)(t - s)}}{\omega(k)}\cos\left[k(X_t^m - X_s^n + x^m - x^n)\right]\,ds\,dt\,dk\right)\nonumber\\
= \, & \alpha\sum_{m, n}\int\!\!\!\int_0^T\!\!\!\int_0^t\frac{\chi_{\Lambda}(k)e^{-\omega(k)(t - s)}}{\omega(k)}D_u\cos\left[k(X_t^m - X_s^n + x^m - x^n)\right]\,ds\,dt\,dk.
\end{align}
Note how $\cos\left[k(X_t^m - X_s^n + x^m - x^n)\right]$ is of the form $f\left[W(g)\right]$, as
\begin{gather}
\cos\left[k(X_t^m - X_s^n + x^m - x^n)\right] = \cos\left\{\int_0^T\left[k^m 1_{[0, t]}(r) - k^n 1_{[0, s]}(r)\right]\, dX_r + k(x^m - x^n)\right\},
\end{gather}
where $k^l$ is the embedding of the vector $k$ into the null vector with $3N$ coordinates, with its components written in the $3l - 2$, $3l - 1$, and $3l$ positions, that is
\begin{gather}
k^l \equiv (0, \ldots, 0, \underbrace{k_1}_{3l - 2}, \underbrace{k_2}_{3l - 1}, \underbrace{k_3}_{3l}, 0, \ldots, 0).
\end{gather}
Therefore, by means of the formula \eqref{equation.malliavin.prescription.single},
\begin{align}
& \int\!\!\!\int_0^T\!\!\!\int_0^t\frac{\chi_{\Lambda}(k)e^{-\omega(k)(t - s)}}{\omega(k)}D_u\cos\left[k(X_t^m - X_s^n + x^m - x^n)\right]\,ds\,dt\,dk\nonumber\\
= \, & -\int\!\!\!\int_0^T\!\!\!\int_0^t\frac{\chi_{\Lambda}(k)e^{-\omega(k)(t - s)}}{\omega(k)}\sin\left[k(X_t^m - X_s^n + x^m - x^n)\right]\left[k^m 1_{[0, t]}(u) - k^n 1_{[0, s]}(u)\right]\,ds\,dt\,dk\nonumber\\
= \, & -i\int\!\!\!\int_0^T\!\!\!\int_0^t\frac{\chi_{\Lambda}(k)e^{-\omega(k)(t - s)}}{\omega(k)}e^{-ik(X_t^m - X_s^n + x^m - x^n)}\left[k^m 1_{[0, t]}(u) - k^n 1_{[0, s]}(u)\right]\,ds\,dt\,dk,
\end{align}
and this is how we finally arrive at the formula
\begin{gather}
D_u\mathcal{A}_T = -i\alpha\sum_{m, n}\int\!\!\!\int_0^T\!\!\!\int_0^t\frac{\chi_{\Lambda}(k)e^{-\omega(k)(t - s)}}{\omega(k)}e^{-ik(X_t^m - X_s^n + x^m - x^n)}\left[k^m 1_{[0, t]}(u) - k^n 1_{[0, s]}(u)\right]\,ds\,dt\,dk.
\end{gather}

The conditional expectation $E\left(D_u\mathcal{A}_T|\mathcal{F}_u\right)$ may now be computed directly, by writing $e^{ikY_r}$ as the product $e^{ik(Y_r - Y_u)}e^{ikY_u}$ for a 3D Brownian motion $Y$ and $r \geq u$, using the independence of $X_n$ and $X_m$ for different $n$ and $m$, and the Markov property of Brownian motion, obtaining
\begin{align}
-\alpha\sum_{m, n}(2 - \delta_{nm})\int\!\!\!\int_u^T\!\!\!\int_0^t & \left[1 + \Theta(s - u)\right]\chi_{\Lambda}(k)e^{-\omega(k)(t - s)}\omega(k)^{-1}\nonumber\\
& \,\,\,\times e^{-k^2(t - u)/2}e^{-k^2(s - u)_{+}/2}\sin[k(X_u^m - X_{s\wedge u}^n + x^m - x^n)]k^m\,ds\,dt\,dk,
\end{align}
where $\Theta$ is Heaviside's theta function and $\delta_{nm}$ is Kronecker's delta. Furthermore, the angular integration in the variable $k$ may be performed explicitly, obtaining that $E(D_u\mathcal{A}_T|\mathcal{F}_u)$ is equal to
\begin{align}
-4\pi\alpha\sum_{m, n}(2 - \delta_{nm})\int_u^T\!\!\!\int_0^t & \frac{(X_u^m - X_{s\wedge u}^n + x^m - x^n)^m}{|X_u^m - X_{s\wedge u}^n + x^m - x^n|}\left[1 + \Theta(s - u)\right]\nonumber\\
&\times\int_{0}^{\Lambda}e^{-\nu(r)(t - s)}\nu(r)^{-1}e^{-r^2(t - u)/2}e^{-r^2(s - u)_+/2}\nonumber\\
&\qquad\quad\times\varphi(r|X_u^m - X_{s\wedge u}^n + x^m - x^n|)r^3\,dr\,ds\,dt,
\end{align}
where $\nu(r) \equiv \sqrt{r^2 + \mu^2}$ and $\varphi(x) = \left(\sin x - x\cos x\right)/x^2$. From here we find that $E(D_u\mathcal{A}_T|\mathcal{F}_u)^2$ can be bounded from above as
\begin{align}
256\pi^2\alpha^2\sum_{m = 1}^N\left(\sum_{n = 1}^N\int_u^T\!\!\!\int_0^t\!\!\int_0^{\Lambda}e^{-r(t - s)}e^{-r^2(t - u)/2}e^{-r^2(s - u)_{+}/2}|\varphi(r|X_u^m - X_{s \wedge u}^n + x^m - x^n|)|r^2\,dr\,ds\,dt\right)^2,
\end{align}
which we define as
\begin{gather}
256\pi^2\alpha^2\sum_{m = 1}^N\left(\sum_{n = 1}^N\mathcal{C}_{m, n}\right)^2.
\end{gather}
We will now estimate each one of the terms $\mathcal{C}_{m, n}$ from above. We find that, by splitting the $s$-integral into the fragments corresponding to $[0, u]$ and $[u, t]$,
\begin{align}
\mathcal{C}_{m, n} \leq \, & \int_{0}^{\Lambda}\!\!\!\int_0^u\frac{1 - e^{-(r + r^2/2)(T - u)}}{1 + r/2}e^{-r(u - s)}|\varphi(r|X_u^m - X_s^n + x^m - x^n|)|r\,ds\,dr\nonumber\\
& \, + \int_0^{\Lambda}\!\!\!\int_u^T\!\!\int_u^t e^{-r(t - s)}e^{-r^2(t - u)/2}e^{-r^2(s - u)/2}|\varphi(r|X_u^m - X_u^n + x^m - x^n|)|r^2\,ds\,dt\,dr\nonumber\\
\equiv \, & \mathcal{D}_{m, n} + \mathcal{E}_{m, n}.
\end{align}

We shall now estimate $\mathcal{D}$ and $\mathcal{E}$ in two separate lemmas.
\begin{lemma}
For all $\varepsilon > 0$, $0 \leq \phi < 1$, and $1 < \theta < 2$,
\begin{gather}
\mathcal{D}_{m, n} \leq \frac{2^{\phi}\|\varphi\|_{\infty}\Gamma(2 - \phi)}{(1 - \phi)\varepsilon^{1 - \phi}} + 2^{-1/2}\|\varphi(x)x^{\theta/2}\|_{\infty}\int_0^{\infty}\frac{r^{-(\theta - 1)/2}}{1 + r/2}\,dr\left(\int_0^u\frac{1_{[0, \varepsilon]}(u - s)}{|X_u^m - X_s^n + x^m - x^n|^{\theta}}\,ds\right)^{1/2}.
\end{gather}
\end{lemma}
\begin{proof}
Let $\varepsilon > 0$. We have that
\begin{align}
\mathcal{D}_{m, n} \leq \, & \int_0^u\!\!\!\int_0^{\infty}\frac{r e ^{-r(u - s)}}{1 + r/2}|\varphi(r|X_u^m - X_s^n + x^m - x^n|)|\,dr\,ds\nonumber\\
= \, & \int_0^{(u - \varepsilon)_+}\!\!\!\int_0^{\infty}\frac{r e ^{-r(u - s)}}{1 + r/2}|\varphi(r|X_u^m - X_s^n + x^m - x^n|)|\,dr\,ds\nonumber\\
& + \int_{(u - \varepsilon)_+}^u\int_0^{\infty}\frac{r e ^{-r(u - s)}}{1 + r/2}|\varphi(r|X_u^m - X_s^n + x^m - x^n|)|\,dr\,ds.
\label{equation.estimate.D}
\end{align}
We will begin by concentrating on the first term in \eqref{equation.estimate.D}. First we notice $1/(1 + r/2)$ is both bounded above by $1$ and $2/r$, which implies that it is bounded by $2^{\phi}/r^{\phi}$ for all $0 \leq \phi \leq 1$, by interpolation. By using this result, we get that for all $0 \leq \phi < 1$,
\begin{align}
& \int_0^{(u - \varepsilon)_+}\!\!\!\int_0^{\infty}\frac{r e ^{-r(u - s)}}{1 + r/2}|\varphi(r|X_u^m - X_s^n + x^m - x^n|)|\,dr\,ds \leq 2^{\phi}\|\varphi\|_{\infty}\int_0^{(u - \varepsilon)_{+}}\int_0^{\infty}r^{1 - \phi}e^{-r(u - s)}\,dr\,ds\nonumber\\
= \, & 2^{\phi}\|\varphi\|_{\infty}\int_0^{\infty}r^{1 - \phi}e^{-r}\,dr\int_0^{(u - \varepsilon)_{+}}(u - s)^{\phi - 2}\,ds\nonumber\\
= \, & \frac{2^{\phi}\|\varphi\|_{\infty}\Gamma(2 - \phi)}{1 - \phi}\left\{\frac{1}{\left[u - (u - \varepsilon)_+\right]^{1 - \phi}} - \frac{1}{u^{1 - \phi}}\right\} \leq \frac{2^{\phi}\|\varphi\|_{\infty}\Gamma(2 - \phi)}{(1 - \phi)\varepsilon^{1 - \phi}}.
\end{align}
As for the second term in \eqref{equation.estimate.D}, we first note that the small and large $x$-behavior of $\sin x - x\cos x$ show that $\varphi(x)x^a$ is bounded for all $a \in [-1, 1]$ and that, moreover, for no other values of $a$ is $\varphi(x)x^a$ bounded. Therefore, one can estimate $|\varphi(y)|$ as $\|\varphi(x)x^a\|_{\infty}|y|^{-a}$ for each $a \in [-1, 1]$, and by doing this we find
\begin{align}
& \int_{(u - \varepsilon)_+}^u\int_0^{\infty}\frac{re^{-r(u - s)}}{1 + r/2}|\varphi(r|X_u^m - X_s^n + x^m - x^n|)|\,dr\,ds\nonumber\\
\leq \, & \|\varphi(x)x^{a}\|_{\infty}\int_0^{\infty}\frac{r^{1 - a}}{1 + r/2}\int_{(u - \varepsilon)_+}^u e^{-r(u - s)}|X_u^m - X_s^n + x^m - x^n|^{-a}\,ds\,dr\nonumber\\
\leq \, & \|\varphi(x)x^a\|_{\infty}\int_0^{\infty}\frac{r^{1 - a}}{1 + r/2}\left(\int_{(u - \varepsilon)_+}^u e^{-2r(u - s)}\,ds\right)^{1/2}\left(\int_{(u - \varepsilon)_+}^u|X_u^m - X_s^n + x^m - x^n|^{-2a}\,ds\right)^{1/2}\,dr\nonumber\\
\leq \, & 2^{-1/2}\|\varphi(x)x^a\|_{\infty}\int_0^{\infty}\frac{r^{1/2 - a}}{1 + r/2}\,dr\left(\int_{(u - \varepsilon)_+}^u |X_u^m - X_s^n + x^m - x^n|^{-2a}\,ds\right)^{1/2}.
\end{align}
\end{proof}
\begin{lemma}
\begin{equation}
\mathcal{E}_{m, n} \leq 2^{-1}(1 - \delta_{nm})\|\varphi(x)/x\|_1.
\end{equation}
\end{lemma}
\begin{proof}
By applying Fubini's Theorem we find that, for $m \neq n$,
\begin{align}
\mathcal{E}_{m, n} = \, & \int_{0}^{\Lambda}\!\!\!\int_u^T\!\!\int_s^T e^{-r(t - s)}e^{-r^2(t - u)/2}e^{-r^2(s - u)/2}|\varphi(r|X_u^m - X_u^n + x^m - x^n|)|r^2\,dt\,ds\,dr\nonumber\\
= \, & \int_0^{\Lambda}\!\!\!\int_u^T\!\!\int_s^T e^{-(r + r^2/2)(t - s)}e^{-r^2(s - u)}|\varphi(r|X_u^m - X_u^n + x^m - x^n|)|r^2\,dt\,ds\,dr\nonumber\\
\leq \, & \int_0^{\infty}\!\!\!\int_u^{\infty}\!\!\int_s^{\infty} e^{-(r + r^2/2)(t - s)}e^{-r^2(s - u)}|\varphi(r|X_u^m - X_u^n + x^m - x^n|)|r^2\,dt\,ds\,dr\nonumber\\
= \, & \int_0^{\infty}\frac{|\varphi(r|X_u^m - X_u^n + x^m - x^n|)|}{r(1 + r/2)}\,dr \leq \int_0^{\infty}|\varphi(r|X_u^m - X_u^n + x^m - x^n|)|r^{-1}\,dr\nonumber\\
= \, & 2^{-1}\|\varphi(x)/x\|_1.
\end{align}
For $n = m$ we simply note that everything vanishes, since $\varphi(0) = 0$.
\end{proof}
We now conclude from the previous two lemmas that
\begin{align}
\mathcal{C}_{m, n} \leq \, & 2^{-1}(1 - \delta_{nm})\|\varphi(x)/x\|_1 + \frac{2^{\phi}\|\varphi\|_{\infty}\Gamma(2 - \phi)}{(1 - \phi)\varepsilon^{1 - \phi}}\nonumber\\
\, & + 2^{-1/2}\|\varphi(x)x^{\theta/2}\|_{\infty}\int_0^{\infty}\frac{r^{-(\theta - 1)/2}}{1 + r/2}\,dr\left(\int_0^u\frac{1_{[0, \varepsilon]}(u - s)}{|X_u^m - X_s^n + x^m - x^n|^{\theta}}\,ds\right)^{1/2}\nonumber\\
\equiv \, & (1 - \delta_{nm})C + D_{\varepsilon} + F_{\theta}\left(\int_0^u\frac{1_{[0, \varepsilon]}(u - s)}{|X_u^m - X_s^n + x^m - x^n|^{\theta}}\,ds\right)^{1/2},
\end{align}
It follows then that
\begin{align}
\frac{E(D_u\mathcal{A}_T|\mathcal{F}_u)^2}{256\pi^2\alpha^2} & = \sum_{m = 1}^N\left(\sum_{n = 1}^N\mathcal{C}_{m, n}\right)^2 \leq N\sum_{m = 1}^N\sum_{n = 1}^N\mathcal{C}_{m, n}^2\nonumber\\
& \leq 2N^2(N - 1)(C + D_{\varepsilon})^2 + 2N^2 D_{\varepsilon}^2 + 2 N F_{\theta}^2\sum_{m, n}\int_0^u\frac{1_{[0, \varepsilon]}(u - s)}{|X_u^m - X_s^n + x^m - x^n|^{\theta}}\,ds.
\label{equation.estimate.malliavin}
\end{align}
\end{subsection}
\begin{subsection}{The Lower Bound}
By using the estimates \eqref{equation.clark.ocone.estimate}, \eqref{equation.estimate.malliavin} we obtain
\begin{align}
E\left(e^{\mathcal{A}_T}\right) & \leq e^{E(\mathcal{A}_T)}E\left[\exp\left(\frac{p^2}{2(p - 1)}\int_0^T\! E(D_t\mathcal{A}_T|\mathcal{F}_t)^2\,dt\right)\right]^{1 - 1/p}\nonumber\\
& \leq e^{E(\mathcal{A}_T)}e^{\gamma T}E\left[\exp\left(\sum_{m, n}\beta\int_0^T\!\!\!\int_0^t\frac{1_{[0, \varepsilon]}(t - s)}{|X_t^m - X_s^n + x^m - x^n|^{\theta}}\,ds\,dt\right)\right]^{1 - 1/p},
\label{equation.estimate.exponential.moment.action.nelson}
\end{align}
where
\begin{align}
\gamma & \equiv 256\pi^2\alpha^2N^2p\left[(N - 1)(C + D_{\varepsilon})^2 + D_{\varepsilon}^2\right] = 256\pi^2\alpha^2N^2 p\left[(N - 1)(C^2 + 2CD_{\varepsilon}) + N D_{\varepsilon}^2\right],\\
\beta & \equiv 256\pi^2\alpha^2 NF_{\theta}^2p^2(p - 1)^{-1},
\end{align}
We will perform now an estimate on the right side of \eqref{equation.estimate.exponential.moment.action.nelson} that is explained and proved in the appendix. It follows from a simple combinatorial argument. It is given by
\begin{align}
& E\left[\exp\left(\beta\sum_{m, n}\int_0^T\!\!\!\int_0^t\frac{1_{[0, \varepsilon]}(t - s)}{|X_t^m - X_s^n + x^m - x^n|^{\theta}}\,ds\,dt\right)\right]\nonumber\\
\leq & \, E\left[\exp\left(N\beta\int_0^T\!\!\!\int_0^t\frac{1_{[0, \varepsilon]}(t - s)}{|Y_t - Y_s|^{\theta}}\,ds\,dt\right)\right]\left\{\sup_{x \in \mathbb{R}^3}E\left[\exp\left(2N\beta\int_0^T\!\!\!\int_0^t\frac{1_{[0, \varepsilon]}(t - s)}{|Y_t - Z_s + x|^{\theta}}\,ds\,dt\right)\right]\right\}^{(N - 1)/2},
\label{equation.combinatorial.argument.used}
\end{align}
where $Y$ and $Z$ are independent 3D Brownian motions. Proceeding from here, we use \cite[Theorems 2.2 and 2.3]{BT}, to find
\begin{align}
E\left[\exp\left(\lambda\int_0^T\!\!\!\int_0^t\frac{1_{[0, \varepsilon]}(t - s)}{|Y_t - Y_s|^{\theta}}\,ds\,dt\right)\right] & \leq \exp\left[\left(A_{\theta}\lambda^{2/(2 - \theta)}\varepsilon^{2/(2 - \theta)} + \frac{B_{\theta}\varepsilon^{1 - \theta/2}\lambda}{1 - \theta/2}\right)T\right],\label{equation.upper.bounds.1}\\
\sup_{x \in \mathbb{R}^3}E\left[\exp\left(\lambda\int_0^T\!\!\!\int_0^t\frac{1_{[0, \varepsilon]}(t - s)}{|Y_t - Z_s + x|^{\theta}}\,ds\,dt\right)\right] & \leq \exp\left[2^{-\theta/(2 - \theta)}A_{\theta}\lambda^{2/(2 - \theta)}\varepsilon^{2/(2 - \theta)}T + o(T)\right],
\label{equation.upper.bounds.2}
\end{align}
where $A_{\theta}$ and $B_{\theta}$ are the explicit functions of $\theta$
\begin{gather}
A_{\theta} \equiv \frac{2^{(3\theta - 2)/(2 - \theta)}\theta^{\theta/(2 - \theta)}(2 - \theta)}{(3 - \theta)^{2\theta/(2 - \theta)}}, \quad B_{\theta} \equiv \frac{\theta\Gamma[(3 - \theta)/2]}{2^{\theta/2}\Gamma(3/2)}.
\end{gather}
From here, we find that
\begin{align}
E(e^{\mathcal{A}_T}) \leq \, & e^{E(\mathcal{A}_T)}e^{\gamma T}\exp\left[\left(A_{\theta}N^{2/(2 - \theta)}\beta^{2/(2 - \theta)}\varepsilon^{2/(2 - \theta)} + \frac{B_{\theta}\varepsilon^{1 - \theta/2}N\beta}{1 - \theta/2}\right)\frac{(p - 1)T}{p}\right]\nonumber\\
& \times\exp\left[A_{\theta}N^{2/(2 - \theta)}(N - 1)\beta^{2/(2 - \theta)}\varepsilon^{2/(2 - \theta)}(p - 1)p^{-1}T + o(T)\right]\nonumber\\
= \, & e^{E(\mathcal{A}_T)}e^{\gamma T}\exp\left[\left(A_{\theta}N^{(4 - \theta)/(2 - \theta)}\beta^{2/(2 - \theta)}\varepsilon^{2/(2 - \theta)} + \frac{B_{\theta}\varepsilon^{1 - \theta/2}N\beta}{1 - \theta/2}\right)\frac{(p - 1)T}{p} + o(T)\right],
\end{align}
and this is how we finally conclude that, from the Feynman-Kac formula \eqref{equation.ground.state.energy},
\begin{align}
& \, E_{\alpha, \mu}^{N, \Lambda} + Q_{\alpha, \mu}^{N, \Lambda}\nonumber\\
\geq & -2^8\pi^2 p\alpha^2N^2(N - 1)(C^2 + 2CD_{\varepsilon}) - 2^8\pi^2 p\alpha^2N^3 D_{\varepsilon}^2\nonumber\\
& - 2^{16/(2 - \theta)}\pi^{4/(2 - \theta)}A_{\theta}N^{(6 - \theta)/(2 - \theta)}\alpha^{4/(2 - \theta)}F_{\theta}^{4/(2 - \theta)}p^{(2 + \theta)/(2 - \theta)}(p - 1)^{-\theta/(2 - \theta)}\varepsilon^{2/(2 - \theta)}\nonumber\\
& -2^8\pi^2\alpha^2 B_{\theta}\varepsilon^{1 - \theta/2}N^2 p F_{\theta}^2(1 - \theta/2)^{-1}.
\label{equation.lower.bound}
\end{align}
We notice that from here the case $\mu = 0$ can be derived by taking $\mu \to 0$. (See the discussion right before \cite[Equation (1.2)]{GM}.) There are several parameters involved in \eqref{equation.lower.bound}, and fully optimizing the lower bound in all of them is beyond the scope the article. Many particular estimates may be derived from \eqref{equation.lower.bound}, however. For example, we may first pick $\theta = 3/2$ and then choose $\varepsilon = N^{-2}\alpha^{-2}$. The lower bound then becomes
\begin{gather}
E_{\alpha, \mu}^{N, \Lambda} + Q_{\alpha, \mu}^{N, \Lambda} \geq -p\alpha^2 N^2(N - 1)(C^2 + 2CD_{\varepsilon}) - p\alpha^2N^3 D_{\varepsilon}^2 - LN - M\alpha^{3/2}N^{3/2},
\end{gather}
where $L$ and $M$ are constants. Even though the $N$ and $\alpha$ behavior of the last two terms is clear, for the first two there is still a degree of freedom given by $\varepsilon$. In particular, the second term may be bounded above as
\begin{gather}
U\alpha^2 N^3\frac{(N^4\alpha^4)^{1 - \phi}}{(1 - \phi)^2},
\end{gather}
where $U$ is a constant. Assuming now that $N^2\alpha^2 \geq e$ we may select $1 - \phi = 1/\log(N^2\alpha^2)$. The expression then becomes
\begin{gather}
4Ue^2\alpha^2 N^3\log^2\left(\alpha N\right),
\end{gather}
where we used the fact that $x^{1/\log x}$ is $e$ for all $x > 0$. A similar analysis may be applied to the first term, leading to a bound of a similar form, except that the logarithmic term is now raised to the first power. After all these computations we conclude that
\begin{gather}
E_{\alpha, \mu}^{N, \Lambda} + Q_{\alpha, \mu}^{N, \Lambda} \geq -D\alpha^2N^3\log^2\left(\alpha N\right),
\end{gather}
for large enough $\alpha$, where $D$ is an explicit positive constant. Now, if $\alpha$ is sufficiently small, it is easy to deduce a lower bound of the form 
$-D\alpha^2 N^3\log^2 N$ when $N \geq 2$, and $-D\alpha^2$ when $N = 1$, by performing a similar procedure (one can choose $\varepsilon = N^{-2}$ when $N \geq 2$, for instance).
\end{subsection}
\end{section}
\begin{section}{Calculations for the Polaron Model}
\label{section.polaron}
We illustrate the previous techniques further in the polaron model of Herbert Fr\"{o}hlich \cite{Fr}. It describes the interaction of a non-relativistic electron with the optical phonon modes of a polar crystal. As the electron moves inside the crystal, it distorts the atom lattice locally, and this distorsion can be represented through quantized waves, called phonons. The polaron model is simpler to handle at the level of functional integrals than the Nelson model, and furthermore, does not require renormalization when the ultraviolet cutoff is removed. For this reason, the calculations in this section will be simpler than before, and some of them will be omitted, in the understanding that they follow from modifying certain arguments in the previous section accordingly. When there are $N$ electrons in the crystal the model has as Hamiltonian
\begin{gather}
H_{\alpha}^{N, \Lambda} = -\sum_{n = 1}^N\frac{\Delta_n}{2} + \int_{\mathbb{R}^3}\chi_{\Lambda}(k)a_k^{\dagger}a_k\,dk + \frac{\sqrt{\alpha}}{2^{3/4}\pi}\sum_{n = 1}^N\int_{\mathbb{R}^3}\frac{\chi_{\Lambda}(k)}{|k|}(e^{ikx_n}a_k + e^{-ikx_n}a_k^{\dagger})\,dk,
\label{equation.polaron.hamiltonian}
\end{gather}
acting on $L^2(\mathbb{R}^{3N})\otimes\mathcal{F}$, where $\mathcal{F}$ is the Fock space over $L^2(\mathbb{R}^3)$. As in the Nelson model, since $\chi_{\Lambda}(k)/|k|$ is an $L^2$ function, $H_{\alpha}^{N, \Lambda}$ is self-adjoint and bounded below \cite{N2}. (The reader should compare at this point \eqref{equation.polaron.hamiltonian} with \eqref{equation.nelson.hamiltonian}.) We do not treat the electrons as fermions, and make the unphysical assumption that the electrons do not repel each other. As for the Nelson model, one can derive an estimate for the ground state energy of $H_{\alpha}^{N, \Lambda}$ involving a Feynman-Kac-like formula, by following the method presented in the appendix, which involves the obtention of a path-integral representation of the kernel of $e^{-TH_{\alpha}^{N, \Lambda}}$. This representation is due essentially to Feynman, who directly integrated the quantum field variables \cite{F1, F2} (i.e. the Fock space), and was later rederived by Nelson by noting that the field variables are driven by an Olstein-Uhlenbeck process when computing matrix elements \cite{N1}. See also the Ph.D. thesis of the author for a direct integration of the phonon field using the Trotter product formula \cite{BPT}. In any case, one obtains the formula
\begin{gather}
\text{inf spec }H_{\alpha}^{N, \Lambda} \geq -\limsup_{T \to \infty}T^{-1}\log\left\{\sup_{x \in \mathbb{R}^3}E\left[\exp\left(\mathcal{A}_T^x\right)\right]\right\},
\end{gather}
with
\begin{gather}
\mathcal{A}_T^x \equiv \frac{\alpha}{2^{3/2}\pi^2}\sum_{m, n}\int_0^T\!\!\!\int_0^t\!\!\int_{|k| \leq \Lambda}\frac{e^{-(t - s)}}{k^2}e^{-ik(X_t^m - X_s^n + x^m - x^n)}\,dk\,ds\,dt.
\end{gather}

We start with the computation of the expectation. It can be carried out easily, in complete analogy to that of the Nelson model action, and one encounters that $E(\mathcal{A}_T)$ (omitting the superscript $x$) can be bounded from above as
\begin{align}
& \sqrt{2}\alpha\pi^{-1}N\int_0^T\!\!\!\int_0^t e^{-(t - s)}\int_0^{\Lambda}e^{-r^2(t - s)/2}\,dr\,ds\,dt\nonumber\\
& + \, \sqrt{2}\alpha\pi^{-1} N(N - 1)\int_0^T\!\!\!\int_0^t e^{-(t - s)}\int_0^{\Lambda}e^{-r^2(t + s)/2}\,dr\,ds\,dt\nonumber\\
\leq & \, \alpha\pi^{-1/2}N\int_0^T\!\!\!\int_0^t\frac{e^{-(t - s)}}{\sqrt{t - s}}\,ds\,dt + \alpha\pi^{-1/2}N(N - 1)\int_0^T\!\!\!\int_0^t \frac{e^{-(t - s)}}{\sqrt{t + s}}\,ds\,dt\nonumber\\
\leq & \, \alpha \pi^{-1/2}N T\int_0^{\infty}\frac{e^{-t}}{\sqrt{t}}\,dt + \alpha\pi^{-1/2}N(N - 1)\int_0^T e^{s}\int_s^T\frac{e^{-t}}{\sqrt{t}}\,dt\,ds.
\label{equation.expectation.polaron}
\end{align}
The last term can be easily seen to sublinear in $T$, as a simple integration by parts shows:
\begin{gather}
\int_0^T e^{s}\int_s^T\frac{e^{-t}}{\sqrt{t}}\,dt\,ds = -\int_0^T\frac{e^{-t}}{\sqrt{t}}\,dt + 2\sqrt{T};
\end{gather}
and since $\int_0^{\infty}t^{-1/2}e^{-t}\,dt = \Gamma(1/2) = \pi^{1/2}$, we obtain
\begin{gather}
E(\mathcal{A}_T) \leq \alpha N T + o(T).
\end{gather}

The computation of $E(D_u\mathcal{A}_T|\mathcal{F}_u)$ is also entirely analogous to that for the Nelson model, and for this reason we shall not perform it explicitly, but merely state the final result, that can be obtained as before, mutatis mutandis:
\begin{align}
-\frac{\sqrt{2}\alpha}{\pi}\sum_{m, n}(2 - \delta_{m, n})&\int_u^T\!\!\!\int_0^t\frac{(X_u^m - X_{s\wedge u}^n + x^m - x^n)^m}{|X_u^m - X_{s\wedge u}^n + x^m - x^n|}\left[1 + \Theta(s - u)\right]\nonumber\\
&\qquad\quad\times\int_0^{\Lambda}e^{-(t - s)}e^{-r^2(t - u)/2}e^{-r^2(s - u)_+/2}\varphi(r|X_u^m - X_{s\wedge u}^n + x^m - x^n|)r\,dr\,ds\,dt.
\end{align}
From here we obtain, easily,
\begin{gather}
E(D_u\mathcal{A}_T|\mathcal{F}_u)^2 \leq \frac{2\alpha^2}{\pi^2}\sum_{m = 1}^N\left(\sum_{n = 1}^N\mathcal{M}_{m, n}\right)^2,
\end{gather}
with
\begin{align}
\mathcal{M}_{m, n} \equiv (2 - \delta_{m, n})&\int_u^T\!\!\!\int_0^t\left[1 + \Theta(s - u)\right] e^{-(t - s)}\nonumber\\
&\qquad\quad\times\int_0^{\Lambda}e^{-r^2(t - u)/2}e^{-r^2(s - u)_{+}/2}|\varphi(r|X_u^m - X_{s\wedge u}^n + x^m - x^n|)|r\,dr\,ds\,dt.
\end{align}
$\mathcal{M}_{m, n}$ is in fact a quantity that can be bounded by a number. To see this, we split $\mathcal{M}_{m, n}(2 - \delta_{m, n})^{-1}$ as before, into an $[0, u]$-integration and an $[u, t]$-one. The first integral yields
\begin{align}
& \int_u^T\!\!\!\int_0^u e^{-(t - s)}\int_0^{\Lambda} e^{-r^2(t - u)/2}|\varphi (r|X_u^m - X_s^n + x^m - x^n|)|r\,dr\,ds\,dt\nonumber\\
= & \, \int_0^{\Lambda}r\int_u^T e^{-(1 + r^2/2)(t - u)}\,dt \int_0^u e^{-(u - s)} |\varphi(r|X_u^m - X_s^n + x^m - x^n|)|\,ds\,dr\nonumber\\
\leq & \, 2\int_0^u e^{-(u - s)}\int_0^{\infty}r^{-1}|\varphi(r|X_u^m - X_s^n + x^m - x^n|)|\,dr\,ds \leq \|\varphi(x)/x\|_1,
\end{align}
while for the second we get (for $m \neq n$)
\begin{align}
& 2\int_u^T\!\!\!\int_u^t e^{-(t - s)}\int_0^{\Lambda}e^{-r^2(t - u)/2}e^{-r^2(s - u)/2}|\varphi(r|X_u^m - X_u^n + x^m - x^n|)|r\,dr\,ds\,dt\nonumber\\
= \, & 2\int_0^{\Lambda}r|\varphi(r|X_u^m - X_u^n + x^m - x^n|)|\int_u^T e^{-r^2(s - u)}\int_s^T e^{-(1 + r^2/2)(t - s)}\,dt \,ds\,dr\nonumber\\
\leq \, & 2\int_0^{\infty}\frac{|\varphi(r|X_u^m - X_u^n + x^m - x^n|)|}{r(1 + r^2/2)}\,dr \leq \|\varphi(x)/x\|_1,
\end{align}
from which we conclude that, for all $m, n$,
\begin{gather}
\mathcal{M}_{m, n} \leq (4 - 3\delta_{mn})\|\varphi(x)/x\|_1,
\end{gather}
and therefore
\begin{align}
E(D_u\mathcal{A}_T|\mathcal{F}_u)^2 \leq \frac{2\alpha^2}{\pi^2}\|\varphi(x)/x\|_1^2 N(4N - 3)^2.
\end{align}
A numerical evaluation then shows that $E(D_u\mathcal{A}_T|\mathcal{F}_u)^2 \leq 0.76\alpha^2 N(4N - 3)^2$. From here we conclude that
\begin{gather}
\text{inf spec }H_{\alpha}^{N, \Lambda} \geq - \alpha N - 0.76\alpha^2 N(4N - 3)^2.
\label{equation.ground.state.energy.polaron}
\end{gather}

There are three important remarks in order, concerning \eqref{equation.ground.state.energy.polaron}. One is that a lower bound for the ground-state energy of the polaron Hamiltonian without an ultraviolet cutoff can be obtained by taking the limit $\Lambda \to \infty$ on the left side of \eqref{equation.ground.state.energy.polaron}, as in the Nelson model (see the corresponding discussion for that case, above, by equation \eqref{equation.main.result}). Second, proceeding formally, a substantially better estimate can be obtained for the case without cutoffs if one starts directly from the action $\mathcal{A}_T$ when $\Lambda = \infty$. Here the resulting action is quite simple, and by means of the combinatorial argument appearing in the appendix, one can split it into a diagonal and off-diagonal part; by then using the results in \cite{BT}, one finds a lower bound for the polaron without cutoffs equal to $-\alpha N - \alpha^2 N^3/4$, which is an extension for $N > 1$ of a result in \cite{BT}. We omit the details, since our main goal here was to provide a lower bound for the Hamiltonian $H_{\alpha}^{N, \Lambda}$, valid for all values of the cutoff $\Lambda$. As the third and final remark, the behavior $N^3$ for large $N$ is consistent with an upper bound for the ground-state energy of the polaron without cutoffs obtained in \cite{BB1}, $-0.109\alpha^2N^3$, which in turn was computed following an idea of Pekar \cite{P}, namely inserting a product trial state into the Hamiltonian $H_{\alpha}^{N, \Lambda}$ with a coherent state for the boson field.
\end{section}
\begin{appendix}
\setcounter{section}{1}
\setcounter{equation}{0}
\begin{section}*{Appendix A: A Feynman-Kac Estimate for the Nelson Model}
Let $H$ be a self-adjoint, bounded-below operator acting on a Hilbert space $\mathcal{H}$. The following formula is easy to prove: for any $\phi$ in $\mathcal{H}$,
\begin{gather}
-\lim_{T \to \infty}T^{-1}\log\left(\phi, e^{-TH}\phi\right) = \inf\left[\supp\left(\mu_{\phi}\right)\right],
\label{equation.infimum.spectral.measure}
\end{gather}
where $\mu_{\phi}$ is the spectral measure associated to $\phi$. A proof can be found in, for example, \cite[Lemma 2.1]{AH}. We shall denote \eqref{equation.infimum.spectral.measure} by $I(\phi)$. From the formula
\begin{gather}
\text{inf spec }H = \inf_{\phi \in S}I(\phi),
\end{gather}
where $S$ is any dense subset of the Hilbert space $\mathcal{H}$, it follows that a universal lower bound on $I(\phi)$ (valid for all $\phi \in S$) will provide a lower bound on the infimum of the spectrum of $H$. We shall focus now on providing an upper bound on $(\phi, e^{-TH}\phi)$ ($\phi \in S$). $H$ will be the massive Nelson model Hamiltonian, as defined previously in this article, and the set $S$ in question we shall choose is the algebraic tensor product $C_0^{\infty}(\mathbb{R}^{3N})\otimes \mathcal{F}$, where $\mathcal{F}$ is the Fock space on $L^2(\mathbb{R}^3)$. (We remark that this is not the Hilbert-space tensor product $\overline{C_0^{\infty}(\mathbb{R}^{3N})\otimes \mathcal{F}}$, which is the closure of the algebraic version. Even though they are typically denoted the same in the literature, the context makes the meaning of the notation clear, so there is usually no risk of confusion between the two.) Any element in $S$ may be written as
\begin{gather}
\sum_{m = 1}^M\psi_m\otimes\left|\xi_m\right>,
\end{gather}
where $M$ is a natural number, $\psi_m$ is in $C_0^{\infty}(\mathbb{R}^{3N})$, and $\left|\xi\right>$ denotes the coherent state of $\mathcal{F}$ defined by the formal relationship $a_k\left|\xi\right> = \xi_k\left|\xi\right>$, for complex numbers $\xi_k$, $k \in \mathbb{R}^3$, with the normalization $\left<0\,\right|\left.\!\xi\right> = 1$. Now, in the Ph.D. thesis of the author a rather complicated formula was derived for $e^{-TH}\psi\otimes\left|\xi\right>$, where $H$ is the Nelson model Hamiltonian, $\psi$ is in $L^2(\mathbb{R}^{3N})$, and $\left|\xi\right>$ is a coherent state of $\mathcal{F}$, as just defined. It is given by
\begin{gather}
e^{-TH}\left[\psi\otimes\left|\xi\right>\right](x) = E^x\left\{\exp\left[\alpha\sum_{n, m}\int_0^T\!\!\!\int_0^t\!\!\int\chi_{\Lambda}(k)\omega(k)^{-1}e^{-ik(X_t^m - X_s^n)}e^{-\omega(k)(t - s)}\,dk\,ds\,dt\right.\right.\nonumber\\
\left.\left. - \sqrt{\alpha}\sum_{n = 1}^N\int_0^T\!\!\!\int\chi_{\Lambda}(k)\omega(k)^{-1/2}e^{-\omega(k)(T - t)}e^{ikX_t^n}\xi_k\,dk\,dt\right]\psi(X_T)\right.\nonumber\\
\otimes\left.\left|-\chi_{\Lambda}\sqrt{\frac{\alpha}{\omega}}\sum_{n = 1}^N\int_0^T e^{-ikX_t^n}e^{-t\omega}\,dt + \xi e^{-T\omega}\right>\right\},
\end{gather}
and therefore $\left(\phi\otimes\left|\eta\right>, e^{-TH}\psi\otimes\left|\xi\right>\right)$ is equal to
\begin{gather}
\int_{\mathbb{R}^3}\overline{\phi(x)}E^x\left\{\exp\left[\alpha\sum_{n, m}\int_0^T\!\!\!\int_0^t\!\!\int\chi_{\Lambda}(k)\omega(k)^{-1}e^{-ik(X_t^m - X_s^n)}e^{-\omega(k)(t - s)}\,dk\,ds\,dt\right.\right.\nonumber\\
\left.\left. - \sqrt{\alpha}\sum_{n = 1}^N\int_0^T\!\!\!\int\chi_{\Lambda}(k)\omega(k)^{-1/2}e^{-\omega(k)(T - t)}e^{ikX_t^n}\xi_k\,dk\,dt\right]\psi(X_T)\right.\nonumber\\
\left<\eta\left|-\chi_{\Lambda}\sqrt{\frac{\alpha}{\omega}}\sum_{n = 1}^N\int_0^T e^{-ikX_t^n}e^{-t\omega}\,dt + \xi e^{-T\omega}\right>\right\}\,dx,
\end{gather}
from which it follows that $\left|\left(\phi\otimes\left|\eta\right>, e^{-TH}\psi\otimes\left|\xi\right>\right)\right|$ is less than or equal to
\begin{align}
\|\psi\|_{\infty}\|\phi\|_{\infty}|\supp{\phi}|\sup_{x \in \mathbb{R}^3}& E^x\left\{\exp\left[\alpha\sum_{n, m}\int_0^T\!\!\!\int_0^t\!\!\int\chi_{\Lambda}(k)\omega(k)^{-1}\cos\left[k(X_t^m - X_s^n)\right] e^{-\omega(k)(t - s)}\,dk\,ds\,dt\right.\right.\nonumber\\
&\left.\left. \qquad\qquad + \, \sqrt{\alpha}N\int_0^T\!\!\!\int\chi_{\Lambda}(k)\omega(k)^{-1/2}e^{-\omega(k)(T - t)}|\xi_k|\,dk\,dt\right]\right.\nonumber\\
&\qquad\qquad\times\left|\left<\eta\left|-\chi_{\Lambda}\sqrt{\frac{\alpha}{\omega}}\sum_{n = 1}^N\int_0^T e^{-ikX_t^n}e^{-t\omega}\,dt + \xi e^{-T\omega}\right>\right|\right\}.
\label{equation.estimate.inner.product}
\end{align}
Now, for two coherent states of the field $\xi$ and $\eta$ (normalized as before) one has the following formula for their inner product,
\begin{gather}
\left(\xi, \eta\right) = \exp\left(\int_{\mathbb{R}^3}\overline{\xi_k}\eta_k\,dk\right),
\end{gather}
and, in particular, $\int_{\mathbb{R}^3}|\xi_k|^2\,dk < \infty$. We now bound some of the terms in \eqref{equation.estimate.inner.product} as follows
\begin{gather}
\int_0^T\!\!\!\int\chi_{\Lambda}(k)\omega(k)^{-1/2}e^{-\omega(k)(T - t)}|\xi_k|\,dk\,dt \leq \left(\int\chi_{\Lambda}(k)\omega(k)^{-3}\,dk\right)^{1/2}\left(\int |\xi_k|^2\,dk\right)^{1/2},
\end{gather}
\begin{align}
& \left|\left<\eta\left|-\chi_{\Lambda}\sqrt{\frac{\alpha}{\omega}}\sum_{n = 1}^N\int_0^T e^{-ikX_t^n}e^{-t\omega}\,dt + \xi e^{-T\omega}\right>\right|\right.\nonumber\\
\leq & \exp\left(\int_{\mathbb{R}^3}N|\eta_k|\chi_{\Lambda}(k)\alpha^{1/2}\omega(k)^{-3/2}\left(1 - e^{-T\omega}\right)\,dk + \int_{\mathbb{R}^3}|\eta_k\xi_k|e^{-T\omega(k)}\,dk\right)\nonumber\\
\leq & \exp\left[N\alpha^{1/2}\left(\int_{\mathbb{R}^3}|\eta_k|^2\,dk\right)^{1/2}\left(\int_{\mathbb{R}^3}\chi_{\Lambda}(k)\omega(k)^{-3}\,dk\right)^{1/2} + \left(\int_{\mathbb{R}^3}|\eta_k|^2\,dk\right)^{1/2}\left(\int_{\mathbb{R}^3}|\xi_k|^2\,dk\right)^{1/2}\right].
\end{align}
From these considerations we conclude that
\begin{align}
& \left(\sum_{m = 1}^M\psi_m\otimes\left|\xi_m\right>, e^{-TH}\sum_{m = 1}^M\psi_m\otimes\left|\xi_m\right>\right) \leq \sum_{m, n = 1}^M\left|\left(\psi_m\otimes\left|\xi_m\right>, e^{-TH}\psi_n\otimes\left|\xi_n\right>\right)\right|\nonumber\\
\leq & \, C_{\psi_1, \ldots, \psi_M, \xi_1, \ldots , \xi_M}\sup_{x \in \mathbb{R}^3}E^x\left\{\exp\left[\alpha\sum_{m, n}\int_0^T\!\!\!\int_0^t\!\!\int\frac{\chi_{\Lambda}(k)}{\omega(k)}\cos\left[k(X_t^m - X_s^n)\right]e^{-\omega(k)(t - s)}\,dk\,ds\,dt\right]\right\},
\end{align}
with $C$ being independent of $T$. From here we conclude that
\begin{align}
& I\left(\sum_{m = 1}^M\psi_m\otimes\left|\xi_m\right>\right)\nonumber\\
\geq & -\limsup_{T \to \infty}T^{-1}\log\left(\sup_{x \in \mathbb{R}^3}E^x\left\{\exp\left[\alpha\sum_{m, n}\int_0^T\!\!\!\int_0^t\!\!\int\frac{\chi_{\Lambda}(k)}{\omega(k)}\cos\left[k(X_t^m - X_s^n)\right]e^{-\omega(k)(t - s)}\,dk\,ds\,dt\right]\right\}\right),
\end{align}
which proves the assertion.
\end{section}
\setcounter{section}{2}
\setcounter{equation}{0}
\begin{section}*{Appendix B: A Simple Combinatorial Argument}
\label{section.combinatorial.argument}
In equation \eqref{equation.combinatorial.argument.used} an estimate was provided on a certain functional integral that resulted from estimating the exponential moment of the Nelson model action. In this appendix we shall proceed to provide a rather general combinatorial argument, from which that particular estimate will follow easily. We will commence by constructing a partition of the set $J_N \equiv \left\{(i, j) : 1 \leq i \leq j \leq N\right\}$ that will be appropriate for our purposes later in this section. Pairs of numbers $(i, j)$ will represent pairs of particles later on. The problem at hand is simple: What is the smallest cardinality of a partition of $J_N$ in which there are no repeated numbers for different pairs in each one of the partition elements? An example of a partition of $J_2$ we would be interested in is $\left\{(1, 1), (2, 2)\right\} \cup \left\{(1, 2)\right\}$, since in the first element there are no repeated numbers for different pairs, nor are there in the second one. On the other hand, $\left\{(1, 1), (1, 2)\right\} \cup \left\{(2, 2)\right\}$ would be discarded, as 1 appears in the two pairs of the first element of the partition. From these considerations, it is clear that the answer to the question posed for $N = 2$ is 2. The answer to the question for any $N$ turns out to be simply $N$, and it is what we will proceed to prove now. We shall call a partition of $J_N$ with the property that its elements have no repeated numbers for different pairs ``admissible.''

\begin{theorem}The smallest cardinality of an admissible partition of $J_N$ is $N$.\end{theorem}
\begin{proof}
We first notice that an admissible partition of $J_N$ must have at least $N$ elements. For $(1, 1)$ must be in one of the partition elements, say $P_1$; $(1, 2)$ ought to be in a partition element $P_2$, with $P_2 \neq P_1$ (for otherwise the partition would be inadmissible); and similarly, $(1, 3)$ belongs to $P_3 \neq P_2, P_1$. By continuing this way, we arrive at a partition element $P_N \neq P_{N - 1}, \ldots , P_1$, and in this manner we see that the partition has at least $N$ elements.

We will now construct an admissible partition of size $N$. As a matter of fact, we will find them all: The admissible partitions of size $N$ will turn out to be in a one-to-one correspondence with the commutatitve latin squares of size $N$. We recall the reader that a commutative latin square of size $N$ is a square of $N \times N$ slots with exactly $N$ numbers appearing without repetitions along rows or columns, and such that the transposed square is equal to the original one. As an example, the additive table of a commutative group of cardinality $N$ is a commutative latin square. To show now the correspondence, consider a commutative latin square $S$, and label its entries as in a matrix, meaning that $S(1, 1)$ is the top leftmost element, $S(1, N)$ is the top rightmost element, etc. In the upper-right triangle $\left\{S(i, j) : 1 \leq i \leq j \leq N\right\}$ all numbers $1, 2, \ldots , N$ appear. If $S$ is now seen as a function from $J_N$ to the set of numbers $\left\{1, 2, \ldots , N\right\}$, the partition $S^{-1}\!\left\{1\right\} \cup \ldots \cup S^{-1}\!\left\{N\right\}$ is admissible. Indeed, a repeated number in two different pairs of an element $S^{-1}\left\{k\right\}$ would violate the very definition of a latin square. To finish the proof we just note that any admissible partition of size $N$ gives rise to a commutative latin square in the obvious way: If the partition is written as $Q_1 \cup \ldots \cup Q_N$ one may color the upper triangular part of an empty $N \times N$ matrix with the number of the partition element the pair $(i, j)$ belongs to: for example, if $(1, 3)$ belongs to $Q_4$, then one may draw a 4 in the slot with row 1 and column 3. One then fills the lower left triangle with the values of the upper right triangle as follows: the slot $(i, j)$ takes the value of the slot $(j, i)$. The resulting matrix is a commutative latin square.
\end{proof}
From now on we shall refer to admissible partitions of minimal size as ``optimal.'' The reader may refer to \cite{HCD} for more information on latin squares in particular and block (or partition) design in general. We will now use the theorem just provided to give a simple proof of the following result, that will imply in particular the estimate mentioned at the beginning of the present appendix, equation \eqref{equation.combinatorial.argument.used}.
\begin{theorem}
\label{theorem.main.combinatorial.result}
Let $\left\{\Gamma_{m, n} : 1 \leq m \leq n \leq N\right\}$ be a collection of non-negative random variables with the following property: for each subset $S$ of $J_N$ with no repeated numbers for different pairs, one has that the elements in $\left\{\Gamma_{m, n} : (m, n) \in S\right\}$ are independent. Then,
\begin{gather}
E\left(\prod_{m \leq n}\Gamma_{m, n}\right) \leq \prod_{m \leq n}E\left(\Gamma_{m, n}^N\right)^{1/N}.
\end{gather}
\label{theorem.combinatorial.argument}
\end{theorem}
\begin{proof}
Let $K$ be an optimal partition of $J_N$, which we write as $J_N = K_1\cup K_2\cup \ldots\cup K_N$, where the $K_i$'s are non-repeating and disjoint. It then follows from H\"{o}lder's inequality for $N$ elements and equal coefficients that
\begin{gather}
E\left(\prod_{m \leq n}\Gamma_{m, n}\right) = E\left(\prod_{i = 1}^N\prod_{(m, n) \in K_i}\Gamma_{m, n}\right) \leq \prod_{i = 1}^N E\left(\prod_{(m, n) \in K_i}\Gamma_{m, n}^N\right)^{1/N}.
\end{gather}
The proof now follows from the fact that $K$ is admissible:
\begin{gather}
\prod_{i = 1}^N E\left(\prod_{(m, n) \in K_i}\Gamma_{m, n}^N\right)^{1/N} = \prod_{i = 1}^N\prod_{(m, n) \in K_i} E\left(\Gamma_{m, n}^N\right)^{1/N} = \prod_{m \leq n}E\left(\Gamma_{m, n}^N\right)^{1/N}.
\end{gather}
\end{proof}
From Theorem \ref{theorem.main.combinatorial.result} equation \eqref{equation.combinatorial.argument.used} follows immediately. By setting $x$ equal to a fixed vector in $\mathbb{R}^{3N}$, and
\begin{align}
\Upsilon_{m, n}(y) \equiv & \exp\left(\beta\int_0^T\!\!\!\int_0^t\frac{1_{[0, \varepsilon]}(t - s)}{{|X_t^m - X_s^n + y|}^{\theta}}\,ds\,dt\right),\\
\Omega_{m, n} \equiv & \, \Upsilon_{m, n}(x^m - x^n),\\
\Gamma_{m, n} \equiv &
\begin{cases}
\Omega_{m, m}\qquad &\text{if $m = n$},\\
\Omega_{m, n}\Omega_{n, m}\qquad &\text{if $m \neq n$},
\end{cases}
\end{align}
we obtain
\begin{align}
& E\left(\prod_{m, n}\Omega_{m, n}\right) = E\left(\prod_{m \leq n}\Gamma_{m, n}\right) \leq \prod_{m = 1}^N E(\Gamma_{m, m}^N)^{1/N}\prod_{m < n}E\left(\Gamma_{m, n}^N\right)^{1/N}\nonumber\\
\leq \, & E\left(\Omega_{1, 1}^N\right)\prod_{m \neq n}E\left(\Omega_{m, n}^{2N}\right)^{1/(2N)} \leq E\left(\Omega_{1, 1}^N\right)\left\{\sup_{y \in \mathbb{R}^3}E\left[\Upsilon_{1, 2}(y)^{2N}\right]\right\}^{(N - 1)/2},
\end{align}
as claimed.
\end{section}
\end{appendix}

\end{document}